\documentclass{article}
\usepackage{amsmath,mathrsfs,amsthm,amssymb,amsfonts,cite,amsopn,mathtools}
\usepackage{listings}
\usepackage{fancyvrb}
\usepackage[latin1]{inputenc}
\usepackage{colortbl}
\usepackage[english]{babel}
\usepackage{times}
\usepackage[normalem]{ulem}
\usepackage[titletoc,title]{appendix}
\usepackage{thmtools, thm-restate}
\usepackage{listings}
\usepackage{algorithm}
\usepackage{algorithmicx}
\usepackage{algpseudocode}
\algnewcommand{\LineComment}[1]{\Statex \(\triangleright\) #1}
\usepackage[hidelinks]{hyperref}
\usepackage{tabularx}
\usepackage{tikz}

\newlength{\defbaselineskip}
\usetikzlibrary{calc,shapes.arrows,chains,positioning}
\usepgflibrary{arrows} 
\usepgflibrary[arrows] 
\usepgflibrary[plotmarks] 
\usetikzlibrary{arrows} 

\DeclareMathOperator{\e}{\bf e}

\DeclareMathOperator{\poly}{poly}

\DeclareMathOperator{\opt}{OPT}
\DeclareMathOperator{\start}{start}

\DeclarePairedDelimiter\ceil{\lceil}{\rceil}
\DeclarePairedDelimiter\dotp{\langle}{\rangle}

\newcommand\myeq{\mathrel{\overset{\makebox[0pt]{\mbox{\normalfont\tiny\sffamily def}}}{=}}}
\newcommand{\maxAi}{\norm[\infty]{A_{:i}}}
\newcommand{\deli}[1]{\nabla_i f_\mu(#1)}

\newtheorem{theorem}{Theorem}[section]
\newtheorem{observation}[theorem]{Observation}

\newtheorem{lemma}[theorem]{Lemma}

\theoremstyle{definition}

\newcommand{\argmax}{\operatornamewithlimits{argmax}}
\newcommand{\argmin}{\operatornamewithlimits{argmin}}

\newcommand{\R}{\mathbb{R}}
\newcommand{\E}{\mathbb{E}}

\newcommand{\vone}{\vec{1}}

\newcommand{\abs}[1]{\ensuremath{\left|#1\right|}}
\newcommand{\norm}[2][]{\ensuremath{\Vert #2 \Vert_{#1}}}

\begin{document}
\title{Unified Acceleration Method for Packing and Covering Problems via Diameter Reduction}
\date{}
\author{
  Di Wang
  \thanks{
  Department of Electrical Engineering and Computer Sciences,
  University of California at Berkeley,
  Berkeley, CA 94720.
  Email: wangd@eecs.berkeley.edu
  }
  \and
  Satish Rao
  \thanks{
  Department of Electrical Engineering and Computer Sciences,
  University of California at Berkeley,
  Berkeley, CA 94720.
  Email: satishr@berkeley.edu
  }
  \and
  Michael W. Mahoney
  \thanks{
  International Computer Science Institute
  and Department of Statistics,
  University of California at Berkeley,
  Berkeley, CA 94720.
  Email:  mmahoney@stat.berkeley.edu
}
}
 \maketitle
\begin{abstract}
The linear coupling method was introduced recently by Allen-Zhu and Orecchia~\cite{AO14} for solving convex optimization problems with first order methods, and it provides a conceptually simple way to integrate a gradient descent step and mirror descent step in each iteration. 
In the setting of standard smooth convex optimization, the method achieves the same convergence rate as that of the accelerated gradient descent method of Nesterov~\cite{Nesterov05}. 
The high-level approach of the linear coupling method is very flexible, and it has shown initial promise by providing improved algorithms for packing and covering linear programs~\cite{AO15-parallel,AO15-stochastic}. 
Somewhat surprisingly, however, while the dependence of the convergence rate on the error parameter $\epsilon$ for packing problems was improved to $O(1/\epsilon)$, which corresponds to what accelerated gradient methods are designed to achieve, the dependence for covering problems was only improved to  $O(1/\epsilon^{1.5})$, and even that required a different more complicated algorithm.
Given the close connections between packing and covering problems and since previous algorithms for these very related problems have led to the same $\epsilon$ dependence, this discrepancy is surprising, and it leaves open the question of the exact role that the linear coupling is playing in coordinating the complementary gradient and mirror descent step of the algorithm.
In this paper, we clarify these issues for linear coupling algorithms for packing and covering linear programs, illustrating that the linear coupling method can lead to improved $O(1/\epsilon)$ dependence for both packing and covering problems in a unified manner, i.e., with the same algorithm and almost identical analysis.
Our main technical result is a novel diameter reduction method for covering problems that is of independent interest and that may be useful in applying the accelerated linear coupling method to other combinatorial problems.
\end{abstract}

\section{Introduction}

A fractional covering problem, in its generic form, can be written as the following linear program (LP):
\begin{equation*}
\min_{x\geq 0}\{c^Tx:Ax\geq b\} ,
\end{equation*}
where $c\in \R^n_{\geq 0},b\in \R^m_{\geq 0}$, and $A\in \R^{m\times n}_{\geq 0}$. 
That is, we want to put weights on the $x_i$-s, for $i\in\left\{1,\ldots,n\right\}$, such that each $j\in\left\{1,\ldots,m\right\}$ is ``covered'' with weight at least $b_j$, where each unit of weight on $x_i$ puts $A_{ij}$ weight on each $j$, and we want to minimize the cost $c^Tx$ in doing so.
Without loss of generality, one can scale the coefficients, in which case one can write this LP in the standard form:
\begin{equation}\label{eq:covLP}
\min_{x\geq 0}\{\vone^Tx:Ax\geq \vone\}  ,
\end{equation}
where $A\in \R^{m\times n}_{\geq 0}$.
The dual of this LP, the fractional packing problem, can be written in this standard form as:
\begin{equation}\label{eq:packLP}
\max_{y\geq 0}\{\vone^Ty:Ay\leq \vone\} .
\end{equation}
We denote by $\opt$ the optimal value of the primal~\eqref{eq:covLP} (which is also the optimal value of the dual~\eqref{eq:packLP}).
In this case, we say that $x$ is a {\it $(1+\epsilon)$-approximation} for the covering LP if $Ax\geq \vone$ and $\vone^Tx\leq (1+\epsilon)\opt$, and we say that $y$ is a {\it $(1-\epsilon)$-approximation} for the packing LP if $Ay\leq \vone$ and $\vone^Ty\geq (1-\epsilon)\opt$.

Packing and covering problems are important classes of LPs with wide applications, and they have long drawn interest in computer science and theoretical computer science. 
Although one can use general LP solvers such as interior point method to solve packing and covering with convergence rate of $\log(1/\epsilon)$, such algorithms usually have very high per-iteration cost, as methods such as the computation of the Hessian and matrix inversion are involved. 
In the setting of large-scale problems, low precision iterative solvers are often more popular choices. 
Such solvers usually run in time with a nearly-linear dependence on the problem size, and they have $\poly(1/\epsilon)$ dependence on the approximation parameter. 
Most such work falls into one of two categories. 
The first category follows the approach of transforming LPs to convex optimization problems, then applying efficient first-order optimization algorithms. 
Examples of work in this category include~\cite{Nemirovski04,AHK12,Nesterov05,Renegar14,AO15-parallel,AO15-stochastic}, and all except~\cite{AO15-parallel,AO15-stochastic} apply to more general classes of LPs. 
The second category is based on the Lagrangian relaxation framework, and some examples of work in this category include~\cite{PST91,Fleischer04,LubyN93,Young01,Young14,KY14}.
For a more detailed comparison of this prior work, see Table $1$ in ~\cite{AO15-stochastic}.
Also, based on whether the running time depends on the width $\rho$, a parameter which typically depends on the dimension and the largest entry of $A$, these algorithms can also be divided into width-dependent solvers and width-independent solvers. 
Width-dependent solvers are usually pseudo-polynomial, as the running time depends on $\rho\opt$, which itself can be large, while width-independent solvers are more efficient in the sense that they provide truly polynomial-time approximation solvers. 

In this paper, we describe a solver for covering LPs of the form~\eqref{eq:covLP}.
The solver is width-independent,%
\footnote{More precisely, our method has a logarithmic dependence on the width, but by Observation~\ref{ob:smallratio} below, this cannot be worse than $\log(nm/\epsilon)$, and thus we consider it as width-independent.} and it is a first-order method with a linear rate of convergence.
That is, if we let $N$ be the number of non-zeros in $A$, then the running time of our algorithm is at worst $O\left(N\frac{\log^2(N/\epsilon)\log(1/\epsilon)}{\epsilon}\right)$.
To simplify the following discussion, we will follow the standard practice of using $\tilde{O}$ to hide poly-log factors, in which case the running time of our algorithm for the covering problem is at worst $\tilde{O}\left(N/\epsilon\right)$.
Among other things, our result is an improvement over the recent bound of $\tilde{O}(N/\epsilon^{1.5})$ provided by Allen-Zhu and Orecchia for the covering problem using a different more complicated algorithm~\cite{AO15-stochastic}, and our result corresponds to the linear rate of convergence that accelerated gradient methods are designed to achieve~\cite{Nesterov05}. 

At least as interesting as the $\tilde{O}(1/\epsilon^{0.5})$ improvement for covering LPs, however, is the context of this problem and the main technical contribution that we developed and exploited to achieve our improvement.
\begin{itemize}
\item
The context for our results has to do with the linear coupling method that was introduced recently by Allen-Zhu and Orecchia~\cite{AO14}.
This is a method for solving convex optimization problems with first order methods, and it provides a conceptually simple way to integrate a gradient descent step and mirror descent step in each iteration. 
In the setting of standard smooth convex optimization, the method achieves the same convergence rate as that of the accelerated gradient descent method of Nesterov~\cite{Nesterov05}, and indeed the former can be viewed as an insightful reinterpretation of the latter. 
The high-level approach of the linear coupling method is very flexible, and it has shown initial promise by providing improved algorithms for packing and covering LPs~\cite{AO15-parallel,AO15-stochastic}. 

The particular motivation for our work is a striking discrepancy between bounds provided for packing and covering LPs in the recent result of Allen-Zhu and Orecchia in~\cite{AO15-stochastic}.
In particular, they provide a $(1-\epsilon)$-approximation solver for the packing problem in $\tilde{O}(N/\epsilon)$, but they are only able to obtain $\tilde{O}(N/\epsilon^{1.5})$ for the covering problem, and for that they need to use a different more complicated algorithm. 
This discrepancy between results for packing and covering LPs is rare, due to the duality between them, and it leaves open the question of the exact role that the linear coupling is playing in coordinating the complementary gradient and mirror descent step of the algorithms for these dual problems. 
\item
Our main technical contribution is a novel diameter reduction method for fractional covering LPs that helps resolve this discrepancy. 
Recall that the smoothness parameter, e.g., Lipschitz constant, and the diameter of the feasible region are the two most natural limiting factors for most gradient based optimization algorithms. 
Indeed, many applications of general first-order optimization techniques can be attributed to the existence of norms or proximal setups for the specific problems that gives both good smoothness and diameter properties. 
In the particular case of coordinate descent algorithms based on the linear coupling idea, we additionally need good coordinate-wise diameter properties to achieve accelerated convergence. 

This is easy to accomplish for packing problems, but it is not easy to do for covering problems, and this is this difference that leads to the $\tilde{O}(1/\epsilon^{0.5})$ discrepancy between packing and covering algorithms in previous work~\cite{AO15-stochastic}.
Our diameter reduction method for general covering problems is straightforward, and it gives both good diameter bounds with respect to the canonical norm for accelerated stochastic coordinate descent (as is needed generally~\cite{Nesterov12,AO15-stochastic}) as well as good coordinate-wise diameter bounds (as is needed for linear coupling~\cite{AO15-stochastic}). 
Thus, it is likely of interest more generally for combinatorial optimization problems.
\end{itemize}
Once the diameter reduction is achieved, the remaining work is mainly straightforward, as we can directly apply known optimization schemes that work well for problems with good diameter properties.
In particular, by using the scheme from~\cite{AO15-stochastic} that was developed for packing LPs, we obtain improved $\tilde{O}\left(N/\epsilon\right)$ results for covering LPs; and this provides a unified acceleration method (unified in the sense that it is with the same algorithm and almost identical analysis) for both packing and covering LPs.

We will start in Section~\ref{sxn:high-level} with a description of some of the challenges in applying acceleration techniques in a unified way to these two dual problems, including those that limited previous work.
Then, in Section~\ref{sxn:diam-reduction} we will present our main technical contribution, a novel diameter reduction method for any covering LP of the form given in~\eqref{eq:covLP}.
Finally, in Section~\ref{sxn:solver} we describe how to combine this with previous work to obtain a unified acceleration method for packing and covering problems.
We include a full description of the latter analysis, with some of the details deferred to Appendix~\ref{App:proofs}.

\section{High-level Description of Challenges}
\label{sxn:high-level}

At a high level, we (as well as Allen-Zhu and Orecchia~\cite{AO15-parallel,AO15-stochastic}) use the same two-step approach of Nesterov~\cite{Nesterov05}. 
The first step involves smoothing, which transforms the constrained problem into a {\it smooth} objective function with trivial or no constraints. 
By smooth, we mean that the gradient of the objective function has some property in the flavor of Lipschitz continuity. 
Once smoothing is accomplished, the second step uses one of several first order methods for convex optimization in order to obtain an approximate solution. 
Examples of standard application of this approach to covering LPs includes the width-dependent solvers of~\cite{Nesterov05,Nemirovski04} as well as multiplicative weights update solvers~\cite{AHK12}.

The first width-independent result following the optimization approach in~\cite{AO15-parallel} achieves width-independence by truncating the gradient, thus effectively reducing the width to $1$. 
The algorithm uses, in a white-box way, the coupling of mirror descent and gradient descent from~\cite{AO14}, which can be viewed as a re-interpretation of Nesterov's accelerated gradient method~\cite{Nesterov05}. 
However, although~\cite{AO15-parallel} uses a coupling of mirror descent and gradient descent, the role of gradient descent is only for width-independence, i.e., to cover the loss incurred by the large component of the gradient (see Eqn.~\eqref{lossregret} below for the precise formulation of this loss), and it is independent of the mirror descent part acting on the truncated gradient. 
In addition,~\cite{AO15-parallel} deviates from the canonical smoothing with entropy, as it instead uses generalized entropy. 
Importantly, the objective function to be minimized is \emph{not} smooth in the standard Lipschitz continuity sense, but it does satisfy a similar local Lipschitz property. 

To improve the sequential packing solver in~\cite{AO15-parallel} with convergence $\tilde{O}(1/\epsilon^3)$ to $\tilde{O}(1/\epsilon)$, the same authors  in~\cite{AO15-stochastic} apply a stochastic coordinate descent method based on the linear coupling idea. 
Barring the difference between Lipschitz and local Lipschitz continuity, the results in~\cite{AO15-stochastic} can be viewed as a variant of accelerated coordinate descent method~\cite{Nesterov12}. 
There are two places where the algorithm achieves an improvement over prior packing-covering results.
\begin{itemize}
\item
One factor of improvement is due to the better coordinate-wise Lipschitz constant over the full dimensional Lipschitz constant. 
Intuitively, in the case of packing or covering, the gradient of variable $x_i$ depends on the penalties of constraints involving $x_i$, which further depend on all the variables in those constraints. 
As a result, if we move all the variables simultaneously, we can only take a small step before changing the gradient of $x_i$ drastically. 
\item
The other factor of improvement comes from accelerating the gradient method. 
The role of gradient descent in the packing solver of~\cite{AO15-stochastic} is twofold. 
First, it covers the loss incurred by the large component of the gradient as in~\cite{AO15-parallel} to give width-independence. 
Second, to accelerate the coupling as in~\cite{AO14}, the gradient descent also needs to cover the regret term incurred by the mirror descent step (see Eqn.~\eqref{lossregret} below for the precise formulation of this regret). 
The adoption of $A$-norm (defined in Eqn.~\eqref{eq:anorm} below) enables the acceleration.
This $A$-norm works particularly well for packing problems, in the sense that it easily leads to good diameter bounds: since the packing constraints impose a naive upper bound of $x^*_i\leq 1/\norm[\infty]{A_{:i}}$ on each variable, thus the feasible region has a small diameter $\max_{x:f(x)\leq f(x_0)}\norm[A]{x-x^*}$. 

The importance of the small diameter is twofold. 
First, the diameter naturally arises in the convergence bound of gradient based methods, so we always need to use a norm or proximal setup giving small diameter to achieve good convergence. 
Second, and more importantly, in this case the small diameter $[0,1/\norm[\infty]{A_{:i}}]$ on each coordinate relates the mirror descent step length and the gradient descent step length. 
As the regret term in mirror descent and the improvement of gradient descent step are both proportional to their respective step lengths, the small coordinate-wise diameter makes it possible to use gradient descent improvement to cover the mirror descent regret.
\end{itemize}
The combination of gradient truncation, stochastic coordinate descent, and acceleration due to small diameter in $A$-norm leads to the $\tilde{O}(N/\epsilon)$ solver for the packing LP~\cite{AO15-stochastic}. 

Shifting to solvers for the covering LP, one obvious obstacle to reproducing the packing result is we no longer have the small diameter in $A$-norm.
Indeed, a naive coordinate-wise upper bound from the covering constraints only gives $x^*_i\leq 1/\min_j\{A_{ji}:A_{ji}> 0\}$. 
Because of this, the covering solver in~\cite{AO15-stochastic} instead use the proximal setup in their earlier work~\cite{AO15-parallel}. 
The particular proximal setup gives a good diameter for the feasible region they use, but it doesn't give a similarly good coordinate-wise diameter to enable the acceleration. 
To improve upon the $O(1/\epsilon^2)$ convergence of standard mirror descent, the authors use a negative-width technique as in~\cite{AHK12} (Theorem $3.3$ with $l=\sqrt{\epsilon}$).
This then leads to the (improved, but still worse than for packing) $\tilde{O}(1/\epsilon^{1.5})$ convergence rate. 
In addition, since they truncate the gradient at a smaller threshold to cover the loss incurred by the large component, they need a more complicated gradient step, leading to a more complicated algorithm than for the packing LP.

To get an $\tilde{O}(1/\epsilon)$ solver for the covering LP, it seems crucial to relate the gradient descent step and mirror descent step the same way as in the packing solver in~\cite{AO15-stochastic}. 
Thus, we will stick with the $A$-norm, and we will work directly to reduce the diameter. 
Our main result (presented next in Section~\ref{sxn:diam-reduction}) is a general diameter reduction method to achieve the same diameter property as in the packing solver, and this enables us (in Section~\ref{sxn:solver}) to extend all the crucial ideas of the packing solver in~\cite{AO15-stochastic}, as outlined in this section, to get a covering solver with running time $\tilde{O}(N/\epsilon)$.

\section{Diameter Reduction Method for General Covering Problems}
\label{sxn:diam-reduction}

Given any covering LP of the form given in~\eqref{eq:covLP}, characterized by a matrix $A$, we formulate an equivalent covering LP with good diameter properties. 
This will involve adding variables and redundant constraints.
We use $i\in [n]$ to denote the indices of the variables (i.e., columns of $A$) and $j\in [m]$ to denote the indices of constraints (i.e., rows of $A$). 
For ease of comparison with~\cite{AO15-stochastic}, and since our unified approach for both packing and covering uses their packing solver and a similar analysis, we use the same notation whenever possible.

For any $i\in [n]$, let 
$$
r_i\myeq\frac{\max_j\{A_{ji}:A_{ji}> 0\}}{\min_j\{A_{ji}:A_{ji}> 0\}} ,
$$
be the ratio between the largest non-zero coefficient and the smallest non-zero coefficient of variable $x_i$ in all constraints, and let $n_i\myeq\ceil*{\log r_i}$.
We first duplicate each original variable $n_i$ times to obtain $\bar{x}_{(i,l)}, i\in [n],l\in [n_i]$ as the new variables. 
In terms of the coefficient matrix, we now have a new matrix, call it $\bar{A}\in \R^{m\times (\sum_i n_i)}_{\geq 0}$, which contains $n_i$ copies of the $i$-th column $A_{:i}$. 
We denote a column of $\bar{A}$ by the tuple $(i,l)$ with $l\in [n_i]$. 
Obviously, the covering LP given by $\bar{A}$ is equivalent to the original covering LP given by $A$. 
Adding additional copies of variables, however, will allow us to improve the diameter.
To reduce the diameter of this new covering LP, we further decrease some of the coefficients in $\bar{A}$, and we put upper bounds on the variables. 
In particular, for $j,i,l$, we have
\begin{equation}
\label{eq:barA}
\bar{A}_{j,(i,l)}=\min\{A_{j,i},2^l\min_j\{A_{ji}:A_{ji}> 0\}\}   ,
\end{equation}
and for variable $\bar{x}_{(i,l)}$, we add the constraint
\begin{equation}\label{eq:upperbound}
\bar{x}_{(i,l)}\leq \frac{2}{2^{l}\min_j\{A_{ji}:A_{ji}> 0\}}   .
\end{equation}

The next lemma shows that the covering LP given by $\bar{A}$ and the covering LP given by $A$ are equivalent.
\begin{lemma}
\label{lemma:equivLP}
Let $\opt$ be the optimal value of the covering LP given by $A$, and let $\overline{\opt}$ be the optimal of the covering LP given by $\bar{A}$ and~\eqref{eq:upperbound}, as constructed above; then $\opt=\overline{\opt}$.
\end{lemma}
\begin{proof}
Given any feasible solution $\bar{x}$, consider the solution $x$ where $x_i=\sum_{l=1}^{n_i}\bar{x}_{(i,l)}$. It is obvious $\vone^Tx=\vone^T\bar{x}$, and $Ax\geq \vone$, as coefficients in $\bar{A}$ are no larger than coefficients in $A$. 
Thus $\opt \leq \overline{\opt}$.

For the other direction, consider any feasible $x$. 
For each $i$, we can assume without loss of generality that 
$$
x_i\leq \frac{1}{\min_j\{A_{ji}:A_{ji}> 0\}}   .  
$$
Let $l_i$ be the largest index such that 
$$
x_{i}\leq \frac{2}{2^{l_i}\min_j\{A_{ji}:A_{ji}> 0\}}   ,
$$ 
and then let
\[
\bar{x}_{(i,l)}=\left\{\begin{array}{lr}
x_i & \text{if }l=l_i\\
0  & \text{if }l\neq l_i \end{array} \right.  .
\]

By construction, $\bar{x}$ satisfies all the upper bounds described in~\eqref{eq:upperbound}. Furthermore, for constraint $j$, we must have $\bar{A}_{j:}\bar{x}\geq 1$. Since for any $i$, $\bar{A}_{j,(i,l_i)}$ differs from $A_{ji}$ only when $A_{ji}>2^{l_i}\min_j\{A_{ji}:A_{ji}> 0\}$, and we must have $l_i<n_i$ in this case by definition of $n_i$, which gives $\bar{x}_{(i,l_i)}=x_{i}\geq \frac{1}{2^{l_i}\min_j\{A_{ji}:A_{ji}> 0\}}$ by our choice of $l_i$ being the largest possible. 
Then we know $\bar{A}_{j,(i,l_i)}=2^{l_i}\min_j\{A_{ji}:A_{ji}> 0\}$, so the $j$-th constraint is satisfied. 
Thus $\opt\geq \overline{\opt}$, and we can conclude $\opt=\overline{\opt}$.
\end{proof}

Given that we have shown that the covering LP defined by $\bar{A}$ and that defined by $A$ are equivalent, we now point out that the seemingly-redundant constraints of~\eqref{eq:upperbound} turn out to be crucial.
The reason is that the feasible region now has a small diameter in the coordinate-wise weighted $2$-norm $\norm[A]{\cdot}$.
In particular, we can rewrite the constraints~\eqref{eq:upperbound} to be 
$$
\bar{x}_{(i,l)}\leq \frac{2}{\norm[\infty]{\bar{A}_{:(i,l)}}}   .
$$ 
For any $i$, this is the same upper bound on $\bar{x}_{(i,l)}$ for $l<n_i$ (consider the row $j^*=\argmax_j\{A_{ji},A_{ji}> 0\}$), and it is a relaxation on $\bar{x}_{(i,n_i)}$.

The price we pay for this diameter improvement is that the new LP defined by $\bar{A}$ is larger than that defined by $A$. 
Two comments on this are in order.
First, by Observation~\ref{ob:smallratio} below, $r_i$ is bounded by $n^2m/\epsilon^2$, and so the diameter reduction step only increases the problem size by $O(\log(mn/\epsilon))$.
Second, we have presented our diameter reduction as an explicit pre-processing step so we can use one unified optimization algorithm (Algorithm~\ref{alg:ascd} below) for both packing and covering, but in practice the diameter reduction would not have to be carried out explicitly. 
It can equivalently be implemented implicitly within the algorithm (a trivially-modified version of Algorithm~\ref{alg:ascd} below) by randomly choosing a scale after picking the coordinate $i$ and then computing $\bar{A}_{j,(i,l)}$ in~\eqref{eq:barA} by shifting bits on the fly. 

Given this reduction, in the rest of the paper, when we refer to the covering LP, we will implicitly be referring to the diameter reduced version, and we have the additional guarantee that there exists an optimal solution $x^*$ to ~\eqref{eq:covLP} such that
\begin{equation}\label{eq:diameter}
0\leq x^*_i \leq \frac{2}{\maxAi} \quad \forall i\in [n]   .
\end{equation}

\section{An Accelerated Solver for (Packing and) Covering LPs}
\label{sxn:solver}


In this section, we will present our solver for covering LPs of the form~\eqref{eq:covLP}.
To motivate this, recall that for packing problems of the form~\eqref{eq:packLP}, bounds of the form~\eqref{eq:diameter} automatically follow from the packing constraints $Ax\leq \vone$. 
For readers familiar with the packing LP solver in~\cite{AO15-stochastic}, it should be plausible that---once we have this diameter property---the same stochastic coordinate descent optimization scheme will lead to a $\tilde{O}(N/\epsilon)$ covering LP solver. 
We now show that indeed the same optimization algorithm for packing LPs can be easily extended to solving covering LPs, thus establishing a unified acceleration method for packing and covering problems.

In Section~\ref{sxn:solver-prelim}, we'll present some preliminaries and describe how we perform smoothing on the original covering objective function; and then
in Section~\ref{sxn:solver-accelerated}, we'll present our main algorithm.
This algorithm involves a mirror descent step, that will be described in Section~\ref{sxn:solver-mirror}, a gradient descent step, that will be described in Section~\ref{sxn:solver-gradient}, and a careful coupling between the two, that will be described in Section~\ref{sxn:solver-coupling}.
Finally, in Section~\ref{sxn:solver-starting}, we will describe how to ensure we start at a good starting point.
Some of the following results are technically-tedious but conceptually-straightforward extensions of analogous results from~\cite{AO15-stochastic}, and some of the results are restated from~\cite{AO15-stochastic}; for completeness, we provide the proof of all of these results, with the latter relegated to Appendix~\ref{App:proofs}.

\subsection{Preliminaries and Smoothing the Objective}
\label{sxn:solver-prelim}

To start, let's assume  that
\[
\min_{j\in [m]}\norm[\infty]{A_{j:}}=1   .
\]
This assumption is without loss of generality: since we are interested in multiplicative $(1+\epsilon)$-approximation, we can simply scale $A$ for this to hold without sacrificing approximation quality.
With this assumption, the following lemma holds.
(This lemma is the same as Proposition $C.2.(a)$ in~\cite{AO15-stochastic}, and its proof is included for completeness in Appendix~\ref{App:proofs}.)
\begin{restatable}{lemma}{optrange}
$\opt\in [1,m]$
\end{restatable}

With $\opt$ being at least $1$, the error we introduce later in the smoothing step will be small enough that the smoothing function approximates the covering LP well enough with respect to $\epsilon$ around the optimum. 

\begin{observation}
\label{ob:smallratio}
Since we are interested in a $(1+\epsilon)$-approximation, then with the above assumption, we can also eliminate the very small and very large entries from the matrix as follows. 
If some entry $A_{ji}\leq \epsilon/(mn)$, then since $\opt\leq m$ we have that $A_{ji}x^*_i\leq \epsilon/n$, and so we can just increase each variable by $\epsilon/n$, in which case we can recover the loss from setting $A_{ji}$ equal to $0$ from the variable in the $j$-th constraint with coefficient at least $1$. 
On the other hand, if some entry $A_{ji}\geq n/\epsilon$, then we can just set variable $i$ to be at least $\epsilon/n$ and ignore constraint $j$. 
Thus, we can eliminate very small and very large entries from the matrix $A$, and we only incur an additional cost of $\epsilon$, but since $\opt\geq 1$, we still obtain a $(1+O(\epsilon))$-approximation. 
\end{observation}


We will turn the covering LP objective into a smoothed objective function $f_\mu(x)$, as used in~\cite{AO15-parallel,AO15-stochastic}, and we are going to find a $(1+\epsilon)$-approximation of the covering LP by approximately minimizing $f_\mu(x)$ over the region 
$$
\Delta\myeq \{x\in \R^n: 0\leq x_i\leq \frac{3}{\maxAi}\}  . 
$$
The function $f_\mu(x)$ is 
\[
f_\mu(x)\myeq \vone^Tx+\max_{y\geq 0}\{y^T(\vone-Ax)+\mu H(y)\}   ,
\]
and it is a smoothed objective in the sense that it turns the covering constraints into soft penalties, with $H(y)$ being a regularization term. 
Here, we use the generalized entropy $H(y)=-\sum_j y_j\log y_j + y_j$, where $\mu$ is the smoothing parameter balancing the penalty and the regularization. 
It is straightforward to compute the optimal $y$, and write $f_\mu(x)$ explicitly, as stated in the following lemma.
\begin{lemma}
$f_\mu(x)=\vone^Tx+\mu \sum_{j=1}^m p_j(x)$, where $p_j(x)\myeq \exp(\frac{1}{\mu}(1-(Ax)_j))$.
\end{lemma}
Optimizing $f_\mu(x)$ over $\Delta$ gives a good approximation to $\opt$, in the following sense. 
If we let $x^*$ be an optimal solution satisfying~\eqref{eq:diameter}, and $u^*\myeq (1+\epsilon/2)x^*\in \Delta$, then we have the properties in the following lemma.
(This lemma is the same as Proposition $C.2$ in~\cite{AO15-stochastic}, and its proof is included for completeness in Appendix~\ref{App:proofs}.)
\begin{restatable}{lemma}{smoothing}\label{lemma:smoothing_properties}
Setting the smoothing parameter $\mu=\frac{\epsilon}{4\log(nm/\epsilon)}$, we have
\begin{enumerate}
\item $f_\mu(u^*)\leq (1+\epsilon)\opt$.
\item $f_\mu(x)\geq (1-\epsilon)\opt$ for any $x\geq 0$.
\item For any $x\geq 0$ satisfying $f_\mu(x)\leq 2\opt$, we must have $Ax\geq (1-\epsilon)\vone$.
\item If $x\geq 0$ satisfies $f_\mu(x)\leq (1+O(\epsilon))\opt$, then $\frac{1}{1-\epsilon}x$ is a $(1+O(\epsilon))$-approximation to the covering LP.
\item The gradient of $f_\mu(x)$ is
\[
\nabla f_\mu(x)=\vone - A^T\vec{p(x)} \quad \text{where} \quad p_j(x)\myeq \exp(\frac{1}{\mu}(1-(Ax)_j) ,
\]
and $\nabla_i f_\mu(x)=1-\sum_j A_{ji}p_j(x)\in [-\infty,1]$.
\end{enumerate}
\end{restatable}

Although $f_\mu(x)$ gives a good approximation to the covering LP, we cannot simply apply the standard (accelerated) gradient descent algorithm to optimize it, as $f_\mu(x)$ doesn't have the necessary Lipschitz-smoothness property. 
However, $f_\mu(x)$ is \emph{locally Lipschitz continuous}, in a sense quantified by the following lemma, and so we have a good improvement with a gradient step within certain range. 
(The following is a ``symmetric'' version%
\footnote{The smoothed objective function for packing LP is $-\vone^Ty+\mu \sum_{j=1}^m q_j(y)$, where $q_j(y)\myeq \exp(\frac{1}{\mu}((Ay)_j-1))$, which is symmetric to $f_\mu(x)$. The properties of $f_\mu(x)$ inherit the symmetry to its packing counterpart, and it can be derived with the same way as~\cite{AO15-stochastic} used for the packing function, but we include it's proof to highlight differences.} 
of Lemma $2.6$ in~\cite{AO15-stochastic}.)

\begin{lemma}\label{lemma:Lipschitz}
Let $L\myeq \frac{4}{\mu}$, for any $x\in \Delta$, and $i\in [n]$
\begin{enumerate}
\item If $\nabla_i f_\mu(x)\in (-1,1)$, then for all $|\gamma|\leq \frac{1}{L\maxAi}$, we have
\[
|\nabla_i f_\mu(x)-\nabla_i f_\mu(x+\gamma \e_i)|\leq L\maxAi |\gamma|  .
\]
\item If $\nabla_i f_\mu(x)\leq -1$, then for all $\gamma \leq \frac{1}{L\maxAi}$, we have
\[
\nabla_i f_\mu(x+\gamma \e_i) \leq (1-\frac{L\maxAi}{2}|\gamma|)\nabla_if_\mu(x)  .
\]
\end{enumerate}
\end{lemma}
\begin{proof}
First, observe the following:
\begin{align*}
\abs{\log \frac{1-\nabla_i f_\mu(x+\gamma\e_i)}{1-\nabla_i f_\mu(x)}}
&= \abs{\int_{0}^{\gamma} - \frac{\nabla_{ii}f_\mu(x+\nu \e_i)}{1-\nabla_if_\mu(x+\nu\e_i)} d\nu}
=\abs{\frac{1}{\mu}\int_{0}^{\gamma}\frac{\sum_j A_{ji}^2p_j(x+\nu\e_i)}{\sum_j A_{ji}p_j(x+\nu\e_i)}d\nu}\\
&\leq \abs{\frac{1}{\mu}\int_{0}^{\gamma} \maxAi d\nu}=\frac{1}{\mu}|\gamma|\maxAi =\frac{L\maxAi}{4}|\gamma|   .
\end{align*}
Then, we have
\[
\exp(-\frac{L\maxAi}{4}|\gamma|)\leq\frac{1-\nabla_i f_\mu(x+\gamma\e_i)}{1-\nabla_i f_\mu(x)}  \leq \exp(\frac{L\maxAi}{4}|\gamma|)   .
\]
Since $\frac{L\maxAi}{4}|\gamma|\leq \frac{1}{4}$ by our assumption, we have $x\leq e^x-1\leq 1.2x$ for $x\in [-\frac{1}{4},\frac{1}{4}]$. 
Thus, it follows that
\[
-\frac{L\maxAi}{4}|\gamma|\leq\frac{\nabla_i f_\mu(x)-\nabla_i f_\mu(x+\gamma\e_i)}{1-\nabla_i f_\mu(x)}  \leq 1.2\frac{L\maxAi}{4}|\gamma|   .
\]
Finally, to prove the lemma we consider the following two cases:
\begin{enumerate}
\item If $\nabla_i f_\mu(x)\in (-1,1)$, then we have
\[
|\nabla_i f_\mu(x)-\nabla_i f_\mu(x+\gamma \e_i)|\leq 1.2(1-\nabla_if_\mu(x))\frac{L\maxAi}{4}|\gamma|\leq L\maxAi |\gamma|   .
\]
\item If $\nabla_i f_\mu(x)\leq -1$, then $1-\nabla_i f_\mu(x)\leq -2\nabla_if_\mu(x)$, and
\begin{align*}
\nabla_i f_\mu(x+\gamma \e_i) &\leq \nabla_i f_\mu(x)+(1-\nabla_if_\mu(x))\frac{L\maxAi}{4}|\gamma| \leq (1-\frac{L\maxAi}{2}|\gamma|)\nabla_if_\mu(x)   .
\end{align*}
\end{enumerate}
\end{proof}

We call $L\maxAi$ the \emph{coordinate-wise local Lipschitz constant}. 
For readers familiar with accelerated coordinate descent method (ACDM)~\cite{Nesterov12}, the $A$-norm is essentially the $\norm[1-\alpha]{\cdot}$ in ACDM~\cite{Nesterov12} with $\alpha=0$, except we use the coordinate-wise local Lipschitz constant instead of the Lipschitz constant to weight each coordinate.
The significance of Lemma~\ref{lemma:Lipschitz} is that for covering LPs the coordinate-wise diameter is  inversely proportional to the coordinate-wise local Lipschitz constant.  
(This fact has been established previously for the case of packing LPs~\cite{AO15-stochastic}.)

\subsection{An Accelerated Coordinate Descent Algorithm}
\label{sxn:solver-accelerated}

\begin{algorithm}
\caption{Accelerated stochastic coordinate descent for both packing and covering}\label{alg:ascd}
\textbf{Input}: $A\in \R^{m\times n}_{\geq 0},x^{\start}\in \Delta,f_\mu,\epsilon$
\textbf{Output}: $y_T\in \Delta$
\begin{algorithmic}[1]
\State $\mu\leftarrow\frac{\epsilon}{4\log(nm/\epsilon)}, L\leftarrow\frac{4}{\mu},\tau\leftarrow\frac{1}{8nL}$
\State $T\leftarrow\ceil*{8nL\log(1/\epsilon)}=\tilde{O}(\frac{n}{\epsilon})$
\State $x_0,y_0,z_0\leftarrow x^{\start},\alpha_0\leftarrow\frac{1}{nL}$
\For{$k=1$ to $T$}
\State $\alpha_k\leftarrow\frac{1}{1-\tau}\alpha_{k-1}$
\State $x_k\leftarrow\tau z_{k-1}+(1-\tau)y_{k-1}$
\State Select $i\in [n]$ uniformly at random.
\LineComment{Gradient truncation:}
\State  Let $\xi^{(i)}_k\leftarrow \left\{\begin{array}{lr}
-1 & \nabla_if_\mu(x_k)<-1\\
\nabla_if_\mu(x_k) & \nabla_if_\mu(x_k)\in [-1,1]\\
1 & \nabla_if_\mu(x_k)>1 \end{array} \right.$
\LineComment{Mirror descent step:}
\State $z_k\leftarrow z^{(i)}_k\myeq \argmin_{z\in \Delta}\{V_{z_{k-1}}(z)+\dotp{z,n\alpha_k\xi^{(i)}_k}\}$.
\LineComment{Gradient descent step:}
\State $y_k\leftarrow y^{(i)}_k\myeq x_k+\frac{1}{n\alpha_k L}(z^{(i)}_k-z_{k-1})$
\EndFor
\State \Return $y_T$.
\end{algorithmic}
\end{algorithm}

We will now show that the accelerated coordinate descent used in packing LP solver in~\cite{AO15-stochastic} also works as a covering LP solver, with appropriately-chosen starting points and smoothed objective functions. 
Consider Algorithm~\ref{alg:ascd}, which is our main accelerated stochastic coordinate descent for both packing and covering.
This algorithm takes as input a matrix $A\in \R^{m\times n}_{\geq 0}$, an initial condition $x^{\start}\in \Delta$, a smoothed function $f_\mu$, and an error parameter $\epsilon$, and it returns as output a vector $y_T\in \Delta$.
The correctness of this algorithm and its running time guarantees for the packing problem have already been nicely presented in~\cite{AO15-stochastic}, and so here we will focus on the covering problem.

Our main result is summarized in the following theorems.

\begin{theorem}\label{thm:main}
With $x^{\start}$ computable in time $\tilde{O}(N)$ to be specified later, Algorithm~\ref{alg:ascd} outputs $y_T$ satisfying $\E[f_\mu(y_T)]\leq (1+6\epsilon)\opt$, and the running time is $\tilde{O}(N/\epsilon)$.
\end{theorem}

\noindent
Given Theorem~\ref{thm:main}, a standard application of Markov bound, together with part $5$ of Lemma~\ref{lemma:smoothing_properties}, gives the following theorem as a corollary.

\begin{theorem}
\label{thm:main2}
There is a algorithm that, with probability at least $9/10$, computes a $(1+O(\epsilon))$-approximation to the fractional covering problem and has $\tilde{O}(N/\epsilon)$ expected running time.
\end{theorem}

\noindent
Not surprisingly, due to the structural similarities of packing and covering problems after diameter reduction, the correctness of Algorithm~\ref{alg:ascd} for covering can be established using the same approach as~\cite{AO15-stochastic} did for packing.
The modifications are fairly straightforward, and we will point out the similarities whenever possible.

Before proceeding with our proof of these theorems, we discuss briefly the optimization scheme from~\cite{AO15-stochastic} we will use. 
First, observe that the $A$-norm, where 
\begin{equation}\label{eq:anorm}
\norm[A]{x}=\sqrt{\sum_{i}\maxAi x_i^2}   ,
\end{equation}
is used as the proximal setup for mirror descent.
The corresponding distance generating function is $w(x)=\frac{1}{2}\norm[A]{x}^2$, and the Bregman divergence is $V_x(y)=\frac{1}{2}\norm[A]{x-y}^2$.\footnote{In particular, $w$ is a $1$-strongly convex function with respect to $\norm[A]{\cdot}$, and $V_x(y)\myeq w(y)-\dotp{\nabla w(x),y-x}-w(x)$. See~\cite{AO14} for a detailed discussion of mirror descent as well as and several interpretations.}

Next, observe that Algorithm~\ref{alg:ascd} works as follows. 
Each iteration integrates a mirror descent step and a gradient descent step. 
The standard analysis of mirror descent gives a convergence of $\frac{1}{\epsilon^2}$, and it depends on the width of the problem. 
Thus, to get a width-independent $\tilde{O}(\frac{N}{\epsilon})$ solver, we need to show that Algorithm~\ref{alg:ascd} addresses both of these issues. 
\begin{itemize}
\item
In order to eliminate the width from the convergence rate, the gradient $\nabla_if_\mu(x_k)$ is split into the small component, 
$\xi^{(i)}_k=\max\{-1,\nabla_if_\mu(x_k)\}\e_i  ,$
and the large component, 
$\eta^{(i)}_k=\nabla_i f_\mu(x_k)\e_i-\xi^{(i)}_k  .$
Only the small component $\xi^{(i)}$ is given to the mirror descent step, and thus the width is effectively $1$. 
However, the truncation incurs loss from the large component, as the mirror descent only acts on the small component. 
Following~\cite{AO15-parallel}, the improvement from the gradient descent step is used to cover that loss.
\item
In order to improve the $1/\epsilon^2$ rate, recall that the $1/\epsilon^2$ in the convergence of mirror descent is largely due to the regret term accumulated along all iterations of mirror descent. 
In order to get to $1/\epsilon$, the improvement from the gradient step also need to cover the regret from the mirror descent step (see Eqn.~\eqref{lossregret} below for the precise formulation of this loss and regret). 
This enables us to telescope both the loss and the regret through all iterations and to bound the total by the gap between $f_\mu(x^{\start})$ and the optimal. 
The remaining terms in the mirror descent also telescope through the algorithm, and they are bounded in total by the distance (in $A$-norm) from $x^{\start}$ to $u^*\in \Delta$. 
\end{itemize}
Then, given these, all we need is an initial condition $x^{\start}$ that is not too far away from the optimal in terms of the function value and not too far away from $u^*$ in $A$-norm. 
For packing, starting with all $0$'s will work.
For covering, we will show later a good enough $x^{\start}$ can be obtained in $\tilde{O}(N)$. 

Finally, here are some lemmas about the algorithm. 
The following two lemmas are invariant to the differences between packing and covering problems, and so they follow directly from the same results in~\cite{AO15-stochastic} (but, for completeness, we include the proofs in Appendix~\ref{App:proofs}). 
The values of parameters $\mu,L,\tau,\alpha_k$ can be found in the description of Algorithm~\ref{alg:ascd}.
The first lemma says that the gradient step we take is always valid (i.e., in $\Delta$), which is crucial in the sense that the gradient descent improvement is proportional to the step length, and we need the step length to be at least $\frac{1}{n\alpha_kL}$ of the mirror descent step length for the coupling to work. 
\begin{restatable}{lemma}{validstep}
We have $x_k,y_k,z_k\in \Delta$ for all $k=0,1,\ldots, T$.
\end{restatable}

\noindent
The second lemma is clearly crucial to achieve the nearly linear time $\tilde{O}(N/\epsilon)$ algorithm.
\begin{restatable}{lemma}{periteration}\label{lemma:periteration}
Each iteration can be implemented in expected $O(N/n)$ time.
\end{restatable}

\subsection{Mirror Descent Step}
\label{sxn:solver-mirror}

We now analyze the mirror descent step of Algorithm~\ref{alg:ascd}:
\[
z_k\leftarrow z^{(i)}_k\myeq \argmin_{z\in \Delta}\{V_{z_{k-1}}(z)+\dotp{z,n\alpha_k\xi^{(i)}_k}\}  .
\]
A lemma of the following form, which here applies to both covering and packing LPs, is needed, and it's proof follows from the textbook mirror descent analysis (or, e.g., Lemma $3.5$ in~\cite{AO15-stochastic}).
\begin{lemma}\label{lemma:mirror}
$\dotp{n\alpha_k\xi^{(i)}_k,z_{k-1}-u^*}\leq n^2\alpha_k^2L\dotp{\xi^{(i)},x_k-y^{(i)}_k}+V_{z_{k-1}}(u^*)-V_{z_k}(u^*)$
\end{lemma}
\begin{proof}
The lemma follows from the following chain of equalities and inequalities.
\begin{align*}
\dotp{n\alpha_k\xi^{(i)}_k,z_{k-1}-u^*}&= \dotp{n\alpha_k\xi^{(i)}_k,z_{k-1}-z_k}+\dotp{n\alpha_k\xi^{(i)}_k,z_{k}-u^*}\\
&=n^2\alpha_k^2L\dotp{\xi^{(i)},x_k-y^{(i)}_k}+\dotp{n\alpha_k\xi^{(i)}_k,z_{k}-u^*}\\
&\leq n^2\alpha_k^2L\dotp{\xi^{(i)},x_k-y^{(i)}_k}+\dotp{-\nabla V_{z_{k-1}}(z^{(i)}_k),z_{k}-u^*}\\
&\leq n^2\alpha_k^2L\dotp{\xi^{(i)},x_k-y^{(i)}_k}+V_{z_{k-1}}(u^*)-V_{z^{(i)}_k}(u^*)-V_{z_{k-1}}(z^{(i)}_k) \\
&\leq n^2\alpha_k^2L\dotp{\xi^{(i)},x_k-y^{(i)}_k}+V_{z_{k-1}}(u^*)-V_{z_k}(u^*)   .
\end{align*}
The first equality follows by adding and subtracting $z_k$, and the second equality comes from the gradient step $y^{(i)}_k= x_k+\frac{1}{n\alpha_k L}(z^{(i)}_k-z_{k-1})$. 
The first inequality is due to the the minimality of $z^{(i)}_k$, which gives
\[
\dotp{\nabla V_{z_{k-1}}(z^{(i)}_k)+n\alpha_k\xi^{(i)}_k,u-z_{k}}\geq 0 \quad \forall u\in\Delta   ,
\]
the second inequality is due to the standard three point property of Bregman divergence, that is $\forall x,y\geq 0$
\[
\dotp{-\nabla V_x(y),y-u}=V_x(u)-V_y(u)-V_x(y)   ,
\]
and the last inequality just drops the term $-V_{z_k}(u^*)$, which is always negative.
\end{proof}

Also, we note that the mirror descent step, defined above in a variational way, can be explicitly written as
\begin{enumerate}
\item $z^{(i)}_k\leftarrow z_{k-1}$
\item $z^{(i)}_k\leftarrow z^{(i)}_k-n\alpha_k\xi^{(i)}_k/\maxAi$
\item If $z^{(i)}_{k,i}<0,z^{(i)}_{k,i}\leftarrow 0$; if $z^{(i)}_{k,i}>3/\maxAi,z^{(i)}_{k,i}\leftarrow 3/\maxAi$.
\end{enumerate}
This is invariant to the difference of packing and covering, and so it follows directly from Proposition $3.6$ in~\cite{AO15-stochastic}. 
It is fairly easy to derive, and so we omit the proof.

\subsection{Gradient Descent Step}
\label{sxn:solver-gradient}

We now analyze the gradient descent step of Algorithm~\ref{alg:ascd}.
In particular, from the explicit formulation of the mirror descent step, we have that $|z^{(i)}_{k,i}-z_{k-1,i}|\leq \frac{n\alpha_k|\xi^{(i)}_k|}{\maxAi}$, which gives 
$$
|y^{(i)}_{k,i}-x_{k,i}|=\frac{1}{n\alpha_kL}|z^{(i)}_{k,i}-z_{k-1,i}|\leq \frac{|\xi^{(i)}_k|}{L\maxAi}   .
$$ 
The gradient step we take is within the local region, and so Lemma~\ref{lemma:Lipschitz} applies. 
We bound the improvement from the gradient descent step in the following lemma, which is symmetric\footnote{The symmetry is between Lemma $2.6$ in~\cite{AO15-stochastic} and Lemma~\ref{lemma:Lipschitz}, as the gradient descent improvement follows directly from the corresponding Lipschitz properties. The actual improvement guarantee is the same as Lemma $3.8$ in~\cite{AO15-stochastic}.} to Lemma $3.8$ in~\cite{AO15-stochastic}.

\begin{lemma}\label{lemma:gradient}
$f_\mu(x_k)-f_\mu(y^{(i)}_k)\geq \frac{1}{2}\dotp{\nabla f_\mu(x_k),x_k-y^{(i)}_k}$
\end{lemma}
\begin{proof}
Since $x_k$ and $y^{(i)}_k$ differ only at coordinate $i$, denote $\gamma=y^{(i)}_{k,i}-x_{k,i}$, we have
\[
f_\mu(x_k)-f_\mu(y^{(i)}_k)=f_\mu(x_k)-f_\mu(x_k+\gamma \e_i)=\int_0^\gamma -\nabla_i f_\mu(x_k+\nu\e_i) d\nu   .
\]
Since $\gamma$ satisfies $|\gamma|\leq \frac{|\xi^{(i)}_k|}{L\maxAi}\leq \frac{1}{L\maxAi}$, we can apply Lemma~\ref{lemma:Lipschitz}.
There are two cases to consider.

If $\nabla_i f_\mu(x_k)\in (-1,1)$, then we have $|\gamma|\leq \frac{|\xi^{(i)}_k|}{L\maxAi} = \frac{|\nabla_if_\mu(x_k)|}{L\maxAi}$, and by Lemma~\ref{lemma:Lipschitz} we have $-\nabla_i f_\mu(x_k+\nu\e_i) \geq -\nabla_if_\mu(x_k)-L\maxAi|\nu|$ in the above integration. 
Thus,
\begin{align*}
f_\mu(x_k)-f_\mu(y^{(i)}_k)&\geq \int_0^\gamma -\nabla_i f_\mu(x_k+\nu\e_i) d\nu\\
&\geq \int_0^\gamma -\nabla_if_\mu(x_k)-L\maxAi|\nu| d\nu\\
&=-\nabla_if_\mu(x_k)\gamma-\frac{L\maxAi}{2}\gamma^2 \\
&\geq -\nabla_if_\mu(x_k)\gamma-\frac{L\maxAi}{2}|\gamma|\frac{|\nabla_if_\mu(x_k)|}{L\maxAi} \\
&=-\frac{1}{2}\dotp{\nabla_if_\mu(x_k),\gamma}=\frac{1}{2}\dotp{\nabla f_\mu(x_k),x_k-y^{(i)}_k}   .
\end{align*}

If $\deli{x_k}\leq -1$, then again by Lemma~\ref{lemma:Lipschitz} we have $-\deli{x_k+\nu\e_i}\geq -(1-\frac{L\maxAi}{2}|\nu|)\deli{x_k}\geq -\frac{1}{2}\deli{x_k}$.
Thus,
\begin{align*}
f_\mu(x_k)-f_\mu(y^{(i)}_k)&\geq \int_0^\gamma -\nabla_i f_\mu(x_k+\nu\e_i) d\nu\\
&\geq \int_0^\gamma -\frac{1}{2}\deli{x_k}d\nu = \frac{1}{2}\dotp{\nabla f_\mu(x_k),x_k-y^{(i)}_k}   .
\end{align*}
\end{proof}

\subsection{Coupling of Gradient and Mirror Descent}
\label{sxn:solver-coupling}

Here, we will analyze the coupling between the gradient descent and mirror descent steps.
This and the next section will give a proof of Theorem~\ref{thm:main}. 

As we take steps on random coordinates, we will write the full gradient as
\[
\nabla f_\mu(x_k)=\E_i[n\deli{x_k}]=\E_i[n\eta^{(i)}_k+n\xi^{(i)}_k]   .
\]
As discussed earlier, we have the small component $\xi^{(i)}_k\in (-1,1)\e_i$ and the large component $\eta^{(i)}_k=\deli{x_k}-\xi^{(i)}_k\in (-\infty,0]\e_i$. 
We put the gradient and mirror descent steps together, and we bound the gap to optimality at iteration $k$:
\begin{align*}
\alpha_k(f_\mu(x_k)-f_\mu(u^*))
\leq& \dotp{\alpha_k\nabla f_\mu(x_k),x_k-u^*}\\
=&\dotp{\alpha_k\nabla f_\mu(x_k),x_k-z_{k-1}}+\dotp{\alpha_k\nabla f_\mu(x_k),z_{k-1}-u^*}\\
=&\dotp{\alpha_k\nabla f_\mu(x_k),x_k-z_{k-1}}+\E_i[\dotp{n\alpha_k\eta^{(i)}_k,z_{k-1}-u^*}+\dotp{n\alpha_k\xi^{(i)}_k,z_{k-1}-u^*}]\\
=&\frac{1-\tau}{\tau}\alpha_k\dotp{\nabla f_\mu(x_k),y_{k-1}-x_k}
  +\E_i[\dotp{n\alpha_k\eta^{(i)}_k,z_{k-1}-u^*}] \\
  &+\E_i[\dotp{n\alpha_k\xi^{(i)}_k,z_{k-1}-u^*}]\\
\leq & \frac{1-\tau}{\tau}\alpha_k(f_\mu(y_{k-1})-f_\mu(x_k))
  +\E_i[\dotp{n\alpha_k\eta^{(i)}_k,z_{k-1}-u^*}]\\
  &+\E_i[n^2\alpha_k^2L\dotp{\xi^{(i)}_k,x_k-y^{(i)}_k}+V_{z_{k-1}}(u^*)-V_{z^{(i)}_k}(u^*)]  .
\end{align*}
The first line is due to convexity. 
The next two lines just break and regroup the terms. 
The fourth line is due to $x_k=\tau z_{k-1}+(1-\tau)y_{k-1}$, so $\tau(x_k-z_{k-1})=(1-\tau)(y_{k-1}-x_k)$. 
The last line is by Lemma~\ref{lemma:mirror}.

We try to use the improvement from the gradient step given in Lemma~\ref{lemma:gradient} to cover the loss from $\eta^{(i)}_k$, and the regret from the mirror descent step:
\begin{equation}\label{lossregret}
\underbrace{\E_i[\dotp{n\alpha_k\eta^{(i)}_k,z_{k-1}-u^*}]}_{\text{loss from }\eta^{(i)}_k}+\underbrace{\E_i[n^2\alpha_k^2L\dotp{\xi^{(i)}_k,x_k-y^{(i)}_k}]}_{\text{regret from mirror descent}}  ,
\end{equation}
and we will use the fact $z_{k-1},z^{(i)}_k,u^*\in \Delta$. 
Consider the following cases.
\begin{enumerate}
\item $\eta^{(i)}_k=0$: In this case, the loss term is $0$. 
We only need to worry about the regret term, and by Lemma~\ref{lemma:gradient}
\[
n^2\alpha_k^2L\dotp{\xi^{(i)}_k,x_k-y^{(i)}_k}\leq 2n^2\alpha_k^2L(f_\mu(x_k)-f_\mu(y^{(i)}_k))   .
\]
\item $\eta^{(i)}_k<0,z^{(i)}_{k,i}<\frac{3}{\maxAi}$: In this case, we increased the $i$-th variable in both the gradient and mirror descent step, and because $z^{(i)}_{k,i}$ is inside $\Delta$ without any projection, we know the step length of gradient descent is exactly $y^{(i)}_{k,i}-x_{k,i}=\frac{1}{n\alpha_kL}\frac{n\alpha_k}{\maxAi}=\frac{1}{L\maxAi}$, together with $z_{k-1}\geq 0$, and $u^*_i\leq \frac{3}{\maxAi}$, we have
\[
\dotp{n\alpha_k\eta^{(i)}_k,z_{k-1}-u^*}\leq \dotp{n\alpha_k\eta^{(i)}_k,-u^*}\leq -n\alpha_k\deli{x_k}\frac{3}{\maxAi} = 3n\alpha_kL\dotp{\nabla f_\mu(x_k),x_k-y^{(i)}_k}   ,
\]
and 
\begin{align*}
\dotp{n\alpha_k\eta^{(i)}_k,z_{k-1}-u^*}+n^2\alpha_k^2L\dotp{\xi^{(i)}_k,x_k-y^{(i)}_k}
\leq &(3n\alpha_kL+n^2\alpha_k^2L)\dotp{\nabla f_\mu(x_k),x_k-y^{(i)}_k}\\
\leq &(6n\alpha_kL+2n^2\alpha_k^2L)(f_\mu(x_k)-f_\mu(y^{(i)}_k))   .
\end{align*}
The last step is by Lemma~\ref{lemma:gradient}.
\item $\eta^{(i)}_k<0,z^{(i)}_{k,i}=\frac{3}{\maxAi}$: In this case, as we know $u^*_i\leq \frac{3}{\maxAi}$, we have
\[
\dotp{n\alpha_k\eta^{(i)}_k,z_{k-1}-u^*}\leq \dotp{n\alpha_k\eta^{(i)}_k,z_{k-1}-z^{(i)}_k}=n^2\alpha_k^2L\dotp{\eta^{(i)}_k,x_{k}-y^{(i)}_k}   ,
\]
and
\begin{align*}
\dotp{n\alpha_k\eta^{(i)}_k,z_{k-1}-u^*}+n^2\alpha_k^2L\dotp{\xi^{(i)}_k,x_k-y^{(i)}_k}
\leq &2n^2\alpha_k^2L\dotp{\nabla f_\mu(x_k),x_k-y^{(i)}_k}\\
\leq &4n^2\alpha_k^2L(f_\mu(x_k)-f_\mu(y^{(i)}_k))   .
\end{align*}
Again, the last step is due to Lemma~\ref{lemma:gradient}.
\end{enumerate}
Since $n\alpha_k<1$ for all $k$, we have in all above cases,
\[
\E_i[\dotp{n\alpha_k\eta^{(i)}_k,z_{k-1}-u^*}]+\E_i[n^2\alpha_k^2L\dotp{\xi^{(i)}_k,x_k-y^{(i)}_k}]
\leq \E_i[8n\alpha_kL(f_\mu(x_k)-f_\mu(y^{(i)}_k))]   .
\]
Back to our earlier derivation, we have
\begin{align*}
\alpha_k(f_\mu(x_k)-f_\mu(u^*))
\leq &\frac{1-\tau}{\tau}\alpha_k(f_\mu(y_{k-1})-f_\mu(x_k))
  +\E_i[\dotp{n\alpha_k\eta^{(i)}_k,z_{k-1}-u^*}]\\
 &+\E_i[n^2\alpha_k^2L\dotp{\xi^{(i)}_k,x_k-y^{(i)}_k}+V_{z_{k-1}}(u^*)-V_{z^{(i)}_k}(u^*)]\\
\leq &\frac{1-\tau}{\tau}\alpha_k(f_\mu(y_{k-1})-f_\mu(x_k))
  +\E_i[8n\alpha_kL(f_\mu(x_k)-f_\mu(y^{(i)}_k)] \\
  &+\E_i[V_{z_{k-1}}(u^*)-V_{z^{(i)}_k}(u^*)]   .
\end{align*}
With our choice of $\tau=\frac{1}{8nL},\alpha_k=\frac{1}{1-\tau}\alpha_{k-1}$, we have
\[
-\alpha_kf_\mu(u^*)\leq 8nL\alpha_{k-1}f_\mu(y_{k-1})-\E_i[8nL\alpha_kf_\mu(y^{(i)}_k)]+\E_i[V_{z_{k-1}}(u^*)-V_{z^{(i)}_k}(u^*)]   .
\]
Telescoping the above inequality along $k=1,\ldots,T$, we get
\[
\E_i[8nL\alpha_Tf_\mu(y_T)]\leq \sum_{k=1}^T\alpha_kf_\mu(u^*) + 8nL\alpha_0f_\mu(y_0)+V_{z_0}(u^*)  ,
\]
and thus
\[
\E_i[f_\mu(y_T)]\leq \frac{\sum_{k=1}^T\alpha_k}{8nL\alpha_T}f_\mu(u^*)+\frac{\alpha_0}{\alpha_T}f_\mu(y_0)+\frac{1}{8nL\alpha_T}V_{z_0}(u^*)   .
\]
We have $\sum_{k=1}^T\alpha_k=\alpha_T\sum_{k=0}^{T-1}(1-\frac{1}{8nL})^k=8nL\alpha_T(1-(1-\frac{1}{8nL})^T)\leq 8nL\alpha_T$, and by our choice of $T=\ceil{8nL\log(1/\epsilon)}$, we also have 
\[
\frac{\alpha_0}{\alpha_T}=(1-\frac{1}{8nL})^T\leq \epsilon,\frac{1}{8nL\alpha_T}\leq \frac{\epsilon}{8nL\alpha_0}= \frac{\epsilon}{8}   ,
\]
and thus
\[
\E_i[f_\mu(y_T)]\leq f_\mu(u^*)+\epsilon f_\mu(y_0)+\frac{\epsilon}{8}V_{z_0}(u^*)   .
\]

\subsection{Finding a Good Starting Point}
\label{sxn:solver-starting}

Here, we will describe how to find a good starting point for the algorithm.
This will permit us to establish the quality-of-approximation and running time guarantees of Theorem~\ref{thm:main}. 

A good starting point $y_0=x^{\start}$ for Algorithm~\ref{alg:ascd} is an initial condition $x^{\start}$ that is not too far away from the optimal in terms of the function value (i.e small $f_\mu(y_0)$), and not too far away from $u^*$ in $A$-norm (i.e. small $V_{z_0}(u^*)$). 
For packing problems, starting with all the all-$0$'s vector will work, but this will not work for covering problems.
Instead, for covering problems, we will show now a good enough $x^{\start}$ can be obtained in $\tilde{O}(N)$. 

To do so, recall that we can get a $2$-approximation $x^\#$ to the original covering LP in time $\tilde{O}(N)$ using various nearly linear time covering solvers, e.g., those of~\cite{KY14,Young14}. 
Without loss of generality, we can assume $x_i^\#\in [0,\frac{2}{\maxAi}]$, since we can use the diameter reduction process as specified in Lemma~\ref{lemma:equivLP} to get a equivalent solution satisfying the conditions. 
Then, we have the following lemma.
\begin{lemma}
Let $x^{\start}=(1+\epsilon/2)x^\#$, we have $x^{\start}\in \Delta$, $f_\mu(x^{\start})\leq 4\opt$, and  $V_{x^{\start}}(u^*)\leq 6\opt$
\end{lemma}
\begin{proof}
It is obvious that $x^{\start}\in \Delta$.
Thus,
\[
\vone^Tx^{\start}=(1+\epsilon/2)\vone^Tx^{\#}\leq (1+\epsilon/2)2\opt \leq 3\opt   .
\]
Furthermore, we have $Ax^{\start}-\vone\geq (1+\epsilon/2)Ax^\#-\vone\geq \frac{\epsilon}{2}\vone$, and so
\[
f_\mu(x^{\start})=\mu\sum_j p_j(x^{\start})+\vone^Tx^{\start}\leq \mu\sum_j\exp(-\frac{\epsilon/2}{\mu})+3\opt\leq \frac{\mu m}{(nm)^2}+3\opt<4\opt   .
\]
For the divergence, we have that
\begin{align*}
V_{x^{\start}}(u^*)=&\frac{1}{2}\sum_i\maxAi (x^{\start}_i-u^*_i)^2\\
= & \frac{1}{2}\sum_i\maxAi ((x^{\start}_i)^2+(u^*)^2_i-2x^{\start}_iu^*_i)\\
\leq & \frac{3}{2}\sum_i x^{\start}_i + u^*_i\\
\leq & \frac{3}{2}(3\opt+\opt) \leq 6\opt  ,
\end{align*}
which proves the lemma.
\end{proof}

\noindent
It is now clear that we have
\[
\E_i[f_\mu(y_T)]\leq f_\mu(u^*)+\epsilon f_\mu(y_0)+\frac{\epsilon}{8}V_{z_0}(u^*)\leq (1+\epsilon)\opt+4\epsilon\opt+\epsilon\opt=(1+6\epsilon)\opt.
\]
Thus, we have the approximation guarantee in Theorem~\ref{thm:main}. The running time follows directly from Lemma~\ref{lemma:periteration} and $T=\tilde{O}(n/\epsilon)$.

\vspace{0.4in}
\noindent
\textbf{Acknowledgments.}
DW was supported by ARO Grant W911NF-12-1-0541, SR was funded by NSF Grant CCF-1118083, and MM acknowledges the support of the NSF, AFOSR, and DARPA.
\vspace{0.1in}

\begin{appendices}

\section{Missing Proofs}
\label{App:proofs}
The following proofs can be found in~\cite{AO15-stochastic}, and we include them here for completeness.
\optrange*
\begin{proof}
By the assumption $\min_{j\in [m]}\norm[\infty]{A_{j:}}=1$, we know at least one constraint has all coefficients at most $1$, so to satisfy that constraint, we must have the sum of the variables to be at least $1$. On the other hand, since each constraint has a variable with coefficient at least $1$ in it, $x=\vone$ clearly satisfies all constraints, so $\opt\leq m$.
\end{proof}
\smoothing*
\begin{proof}
\begin{enumerate}
\item Since $Ax^*\geq \vone$, and $u^*=(1+\epsilon/2)x^*$, we have $(Au^*)_j-1\geq\epsilon/2$ for all $j$. Then $p_j(u^*)\leq \exp(-\frac{1}{\mu}\frac{\epsilon}{2})=(\frac{\epsilon}{mn})^2$, and $f_\mu(u^*)=\vone^Tu^*+\mu\sum_{j=1}^mp_j(u^*)\leq (1+\epsilon/2)\opt+\mu m (\frac{\epsilon}{mn})^2\leq (1+\epsilon)\opt$.
\item By contradiction, suppose $f_\mu(x)<(1-\epsilon)\opt$, since $f_\mu(x)<\opt\leq m$, we must have $p_j(x)<m/\mu$ for any $j$, which implies $(Ax)_j\geq 1-\epsilon$. By definition of $\opt$, we have $\vone^Tx \geq (1-\epsilon)\opt$, since $Ax\geq (1-\epsilon)\vone$. This gives a contradiction as $f_\mu(x)>\vone^Tx\geq (1-\epsilon)\opt$.
\item By contradiction, suppose there is some $j$ such that $(Ax)_j-1\leq -\epsilon$, then as in the last part, we have $\mu p_j(x)\geq \mu (\frac{mn}{\epsilon})^4 >2\opt$, contradicting $f_\mu(x)\leq 2\opt$.
\item For any $x$ satisfying $f_\mu(x)\leq (1+O(\epsilon))\opt\leq 2\opt$, by last part we know $Ax\geq (1-\epsilon)\vone$, so $A(\frac{1}{1-\epsilon}x)\geq \vone$. We also have $\vone^T(\frac{1}{1-\epsilon}x)=\frac{1}{1-\epsilon}\vone^Tx<\frac{1}{1-\epsilon}f_\mu(x)\leq (1+O(\epsilon))\opt$.
\item This is by straightforward computation.
\end{enumerate}
\end{proof}
\validstep*
\begin{proof}
At the start $x_0=y_0=z_0=x^{\start}\in \Delta$ by assumption. $z_k$ is always in $\Delta$ as we take the projection in the mirror descent step. If we can further show $y_k\in \Delta$ for all $k$, we are done, since $x_k$ is a convex combination of $y_{k-1},z_{k-1}$. To show $y_k\in \Delta$, we write $y_k$ as a convex combination of $z_0,\ldots,z_k$, $y_k=\sum_{l=0}^kc^l_kz_l$.
At $k=0$, we have $y_0=z_0$, and at $k=1$, $y_1=x_1+\frac{1}{n\alpha_1L}(z_1-z_0)=\frac{1}{n\alpha_1L}z_1+(1-\frac{1}{n\alpha_1L})z_0$, as $x_1=y_0=z_0$. For $k\geq 2$, we can verify
\[
c^l_k=\left\{\begin{array}{lr}
(1-\tau)c^l_{k-1} & l=0,\ldots,k-2\\
(\frac{1}{n\alpha_{k-1}L}-\frac{1}{n\alpha_kL})+\tau(1-\frac{1}{n\alpha_{k-1}L}) & l=k-1\\
\frac{1}{n\alpha_kL} & l=k \end{array} \right.
\]
since
\begin{align*}
y_k&=x_k+\frac{1}{n\alpha_kL}(z_k-z_{k-1})\\
&=\tau z_{k-1} + (1-\tau)y_{k-1} + \frac{1}{n\alpha_kL}(z_k-z_{k-1})\\
&=\tau z_{k-1} + (1-\tau)(\sum_{l=0}^{k-2}c^l_{k-1}z_l+\frac{1}{n\alpha_{k-1}L}z_{k-1}) + \frac{1}{n\alpha_kL}(z_k-z_{k-1})\\
&=(\sum_{l=0}^{k-2}(1-\tau)c^l_{k-1}z_l)+((\frac{1}{n\alpha_{k-1}L}-\frac{1}{n\alpha_kL})+\tau(1-\frac{1}{n\alpha_{k-1}L}))z_{k-1}+\frac{1}{n\alpha_kL}z_k
\end{align*}
As $\alpha_k\geq \alpha_{k-1}$, and $\alpha_0=\frac{1}{nL}$, we have $c^l_k\geq 0$ for all $l,k$, and it is easy to check the coefficients sum to $1$ for each $k$.
\end{proof}
\periteration*
\begin{proof}
We show how to implement a iteration conditioned on $i$ in time $O(\norm[0]{A_{:i}})$, where $\norm[0]{A_{:i}}$ is the number of non-zeros in column $i$, thus give a expected running time of $O(N/n)$ for each iteration. We maintain the following quantities
\[
z_k\in \R_{\geq 0}^n,az_k\in \R_{\geq 0}^m,y'_k\in \R^n,ay_k\in \R^m, B_{k,1},B_{k,2}\in \R_{+}
\]
with the following invariants always satisfied throughout the algorithm
\begin{equation}\label{eq:invar1}
Az_k=az_k
\end{equation}
\begin{equation}\label{eq:invar2}
y_k=B_{k,1}z_k+B_{k,2}y'_k,\quad Ay_k=B_{k,1}az_k+B_{k,2}ay_k
\end{equation}
When $k=0$, we let $az_k=Az_0,y'_k=y_0,ay_k=Ay_0,B_{k,1}=0,B_{k,2}=1$, and it is clear all the invariants are satisfied. For $k=1,2,\ldots,T$:
\begin{itemize}
\item The step $x_k=\tau z_{k-1}+(1-\tau)y_{k-1}$ does not need to be implemented.
\item Computation of $\nabla_i f(x_k)$ requires the value of $p_j(x_k)=\exp(\frac{1}{\mu}(1-(Ax_k)_j))$ for each $j$ such that $A_{ji}\neq 0$, and we can get the value
\[
(Ax_k)_j=\tau(Az_{k-1})_j+(1-\tau)(Ay_{k-1})_j=(\tau+(1-\tau)B_{k-1,1})(az_{k-1})_j+(1-\tau)B_{k-1,2}ay_{k-1,j}
\]
for each such $j$. This can be computed in $O(1)$ time for each $j$, and $O(\norm[0]{A_{:i}})$ time in total.
\item The mirror descent step $z^{(i)}_k\myeq \argmin_{z\in \Delta}\{V_{z_{k-1}}(z)+\dotp{z,n\alpha_k\xi^{(i)}_k}\}$ is simply $z_k=z_k+\delta \e_i$ where $\delta\in \R$ can be computed in $O(1)$ time. $z_k=z_{k-1}+\delta \e_i$ yields $y_k=\tau z_{k-1}+(1-\tau)y_{k-1}+\frac{\delta}{n\alpha_k L}\e_i$ by the gradient descent step. Therefore, we can update the values accordingly
\[
z_{k}\leftarrow z_{k-1}+\delta \e_i,\quad az_k\leftarrow az_{k-1}+\delta A_{:i}
\] and
\[
\begin{array}{ll}
B_{k,1}\leftarrow \tau+(1-\tau)B_{k-1,1} & B_{k,2}\leftarrow (1-\tau)B_{k-1,2}\\
y'_k\leftarrow y'_{k-1}+\delta(-\frac{B_{k,1}}{B_{k,2}}+\frac{1}{n\alpha_k L}\frac{1}{B_{k,2}}) \e_i & ay_k\leftarrow ay_{k-1}+\delta (-\frac{B_{k,1}}{B_{k,2}}+\frac{1}{n\alpha_k L}\frac{1}{B_{k,2}}) A_{:i}
\end{array}
\]
We can verify that after the updates, the invariants still hold
\begin{align*}
y_k=&B_{k,1}z_k+B_{k,2}y'_k=B_{k,1}(z_{k-1}+\delta \e_i)+B_{k,2}(y'_{k-1}+\delta(-\frac{B_{k,1}}{B_{k,2}}+\frac{1}{n\alpha_k L}\frac{1}{B_{k,2}}) \e_i)\\
=& B_{k,1}z_{k-1}+B_{k,2}(y'_{k-1}+\delta(\frac{1}{n\alpha_k L}\frac{1}{B_{k,2}})\e_i)\\
=& B_{k,1}z_{k-1}+B_{k,2}y'_{k-1}+\frac{\delta}{n\alpha_k L}\e_i\\
=&(\tau+(1-\tau)B_{k-1,1})z_{k-1}+((1-\tau)B_{k-1,2})y'_{k-1}++\frac{\delta}{n\alpha_k L}\e_i\\
=&\tau z_{k-1}+(1-\tau)y_{k-1}++\frac{\delta}{n\alpha_k L}\e_i
\end{align*}
It is also straightforward to verify $Ay_k=B_{k,1}az_k+B_{k,2}ay_k$ equals $Ay_k=\tau Az_{k-1}+(1-\tau)Ay_{k-1}++\frac{\delta}{n\alpha_k L}A\e_i$. The updates are dominated by the updates on $az_k$ and $ay_k$, which take $O(\norm[0]{A_{:i}})$ time.
\end{itemize}
\end{proof}
\end{appendices}
\bibliographystyle{plain}
\bibliography{cover}
\end{document}